%% file: focs.tex
\documentclass[letterpaper,11pt]{article}
\usepackage{amsthm}  
\usepackage{amsmath, amssymb}
\usepackage{color}
\usepackage{fullpage}
\usepackage[numbers]{natbib}
\usepackage[noend]{algorithmic}
\usepackage{thmtools, thm-restate}
\usepackage{tikz}
\usepackage{enumerate}
\usepackage{thm-restate}
\usepackage{graphicx}
\usepackage{amsopn}
\usepackage{caption}
\usepackage{hyperref}
\usepackage{subcaption}
\usepackage{zerosum,csquotes}
\usepackage{enumitem,linegoal}
\usepackage[linesnumbered,ruled,vlined]{algorithm2e}
\usepackage{comment}
\usepackage{float}
\usepackage[super]{nth}
\usepackage{tcolorbox}
\usepackage{bold-extra}
\tcbuselibrary{skins,breakable}

\usepackage{ctable}					

\usepackage{mathtools}
\usepackage[flushleft]{threeparttable}
\usepackage[normalem]{ulem}

\usepackage{footnote}
\makesavenoteenv{tabular}
\makesavenoteenv{table}

\usepackage[margin=1in]{geometry}

 \newtheorem{theorem}{Theorem}[section]
 
  \newtheorem{claim}[theorem]{Claim}

  \newtheorem{mobsv}[theorem]{Observation}

\theoremstyle{definition}
 \newtheorem{definition}[theorem]{Definition}

\makeatletter
\newif\ifqed
\def\GrabProofArgument[#1]{ #1: \egroup\ignorespaces}
\def\proof{\noindent\textbf\bgroup Proof%
	\@ifnextchar[{\GrabProofArgument}{. \egroup\ignorespaces}\global\qedtrue}

\def\qedhere{\ifmmode\tag*{\qedsign}\else\hspace*{\fill}\qedsign\medskip\fi\global\qedfalse}
\def\qedsign{$\Box$}
\makeatother

\input{macros}

\newcommand*\samethanks[1][\value{footnote}]{\footnotemark[#1]}

\newcounter{proccnt}

\newcommand{\konote}[1]{}

\title{Asymmetric Streaming Algorithms for \\ Edit Distance and LCS}
\author{
	Alireza Farhadi\thanks{University of Maryland. Email: \texttt{\{farhadi,hajiagha\}@cs.umd.edu}.}
	\and
	MohammadTaghi Hajiaghayi\samethanks[1]
	\and
	Aviad Rubinstein\thanks{Stanford University. Email: \texttt{aviad@cs.stanford.edu}.}
	\and
	Saeed Seddighin\thanks{TTIC. Email: \texttt{saeedreza.seddighin@gmail.com}.}
}

\newcommand{\Saeed}[1]{}
\begin{document}
	\newcommand{\ignore}[1]{}
\renewcommand{\theenumi}{(\roman{enumi})}
\renewcommand{\labelenumi}{\theenumi.}
\sloppy

%
%

\date{}

\maketitle


\begin{abstract}
\input{abstract}
\end{abstract}
\newpage
\input{intro}

\section{Constant Approximation for Edit Distance}\label{sec:constant}
\input{ed-constant}
\section{$(1-\epsilon)$-Approximation of LCS}
\input{lcs-rootn}
\section{$(1+\epsilon)$-Approximation of ED}\label{sec:rootn}
\input{ed-rootn}

\bibliographystyle{apalike} 
\bibliography{editdistance}
\newpage

\appendix
\section{Omitted proofs}\label{sec:appx}

\input{omitted}

\end{document}

%% file: macros.tex
\newcommand{\OPT}{\ensuremath{\mathsf{OPT}}\xspace}

\newcommand{\opt}{\OPT}

\newcommand{\of}{\bar{s}}
\newcommand{\on}{s}

\newcommand{\ith}[1]{${#1}^{\textnormal{th}}$}
\DeclareMathOperator{\LCS}{\textsf{LCS}}
\DeclareMathOperator{\lcs}{lcs}
\DeclareMathOperator{\ED}{\textsf{ED}}
\DeclareMathOperator{\ed}{ed}
\DeclareMathOperator{\neps}{APPROXIMATE-CLOSEST-SUBSTR}
\DeclareMathOperator{\lcsrootn}{APPROXIMATE-LCS}
\DeclareMathOperator{\edrootn}{APPROXIMATE-ED}

\DeclareMathOperator{\lcsp}{LCSPosition}

\newenvironment{tbox}{
\begin{tcolorbox}[enhanced,
                  boxsep=2pt,
                  left=1pt,
                  right=1pt,
                  top=4pt,
                  boxrule=1pt,
                  arc=0pt,
                  colback=white,
                  colframe=black,
                  breakable]
}{
\end{tcolorbox}
}

\definecolor{mygreen}{RGB}{20,140,80}
\definecolor{mylightgray}{RGB}{230,230,230}

\definecolor{mygreen}{RGB}{20,140,80}
\definecolor{mydarkgray}{gray}{0.15} 
\definecolor{oceanblue}{HTML}{2c55c2}
\hypersetup{
     colorlinks=true,
     citecolor= mygreen,
     linkcolor= black
}

\SetCommentSty{mycommfont}

%% file: abstract.tex
The \textit{edit distance} (\textsf{ED}) and \textit{longest common subsequence} (\textsf{LCS}) are two fundamental problems which quantify how similar two strings are to one
another. In this paper, we consider these problems in the asymmetric streaming model introduced by Andoni \textit{et
al.}~\cite{andoni2010polylogarithmic} (FOCS'10) and Saks and Seshadhri \cite{saks2013space} (SODA'13). In this model we have random access to one string and streaming access the other string. 
Our main contribution is  a constant factor approximation algorithm for \textsf{ED} with the memory of $\tilde O(n^{\delta})$ for any constant $\delta > 0$. In addition to this, we present an upper bound of $\tilde O_\epsilon(\sqrt{n})$ on the memory needed to approximate \textsf{ED} or \textsf{LCS} within a factor $1+\epsilon$. All our algorithms are deterministic and run in a single pass.

For approximating \textsf{ED} within a constant factor, we discover yet another application of triangle inequality, this time in the context of streaming algorithms. Triangle inequality has been previously used to obtain subquadratic time approximation algorithms for \textsf{ED}. Our technique is novel and elegantly utilizes triangle inequality to save memory at the expense of an exponential increase in the runtime.

%% file: intro.tex
\Saeed{sometimes we are writing \textsf{LCS}, sometime we just write LCS. Please make it consistent.}

\section{Introduction}\label{introduction}
We consider \textit{edit distance} (\textsf{ED}) and \textit{longest
common subsequence} (\textsf{LCS}) which are  classic problems measuring the similarity
 between two strings. Edit
distance is defined on two  strings $\on$ and $\of$ and seeks the
smallest number of character insertions, character deletions, and
character substitutions to transform $\on$ into $\of$. While in edit
distance the goal is to make a transformation, longest common
subsequence asks for the largest string that appears as a
subsequence in both $\on$ and $\of$.

Edit distance and longest common subsequence have applications in
various contexts, such as computational biology, text processing,
compiler optimization, data analysis, image analysis, among others.
As a result, both problems have been subject to a plethora of
studies since 1950 (e.g. see
~\cite{bellman1957dynamic,boroujeni2018approximating,abboud2015tight,backurs2015edit,andoni2012approximating,andoni2010polylogarithmic,indyk2001algorithmic,
    batu2006oblivious,crochemore2001fast,
    de1999extensive,hunt1977fast,masek1980faster,
    crochemore2003subquadratic,gusfield1997algorithms,landau1998incremental,
    bringman2018multivariate,
     alves2006coarse,bestpaper,
    saha17, saha16, saha15, 
    saha172, saha19, amitkumarted, rabaniedit,
     lcs-hardness2, lcs-hardness3, hss19, andonied1, smootheded, AK07, andonied3, lcrs-journal,
     dwt, charikar2018estimating, dtw2, saks2013space, aviad1, abboud-lcs-hardness2, spaaedulam, lcsnew,bestpaper}).

Both of the problems are often used to measure the similarity of large strings. For example, a human genome consists of almost three billion base pairs that are modeled as a string for similarity testing. Classic algorithms for the problems require quadratic runtime as well as linear memory to find a solution. Unfortunately, none of these bounds seem practical for real-world applications. Therefore, recent work on \textsf{ED} and \textsf{LCS} focus on obtaining fast algorithms~\cite{saeedfocs19,lcsnew,boroujeni2018approximating,andoni2012approximating,andoni2010polylogarithmic,andonied1,AK07,rubinstein2019reducing,im1,im2,bestpaper} as well as solutions with small memory~\cite{hss19,boroujeni2018approximating,CGK16,DBLP:conf/soda/GopalanJKK07}.

The {\em streaming} setting is an increasingly popular framework to model memory constraints. In this setting, the
input arrives as a data stream while only sublinear memory is available to the algorithm. The goal is to design an algorithm that solves/approximates the solution 
by taking a few passes over the data.
 While several works have studied \textsf{ED} and \textsf{LCS} in the streaming model (see Section \ref{sec:related} for a detailed discussion), positive results are known only for the low-distance regime~\cite{liben2006finding,SW07,belazzougui2016edit,CGK16}. 
In addition to this, strong lower bounds are given for the streaming variant of \textsf{LCS}~\cite{liben2006finding,SW07}.

Inspired by the work of Andoni \textit{et
al.}~\cite{andoni2010polylogarithmic} (FOCS'10), Saks and Seshadhri \cite{saks2013space} (SODA'13) studied the problem of approximating $n\ -$~\textsf{LCS} (which is the edit distance between two strings when insertions and deletions, but not substitutions, are allowed) in the asymmetric model. In this model we have random access to one of the strings and streaming access to the other string. They showed that $(1+\epsilon)$ approximation of $n\ -$~\textsf{LCS} can be found with a memory of $\tilde{O_\epsilon}(\sqrt{n})$. 

In this work, we study \textsf{ED} and \textsf{LCS} in the asymmetric model.
We present a single-pass deterministic constant factor approximation algorithm for \textsf{ED} that uses only $\tilde{O}(n^{\delta})$ memory for any constant $\delta > 0$. In addition to this, we show that with the memory of $ \tilde{O_\epsilon}(\sqrt{n})$ one can approximate both \textsf{ED} and \textsf{LCS} within a factor of $1 \pm \epsilon$. All our algorithms are deterministic and run in a single-pass. Moreover, our algorithm for \textsf{LCS} is tight due to a lower bound given in~\cite{DBLP:conf/focs/GalG07}. It is also worth mentioning that the lower bound of $\Omega(\log^2 n/\epsilon)$ is known for computing $1+\epsilon$ approximation of $n\ -$~\textsf{LCS} due to the result of \cite{ naumovitz2014polylogarithmic}.

\textsf{LIS} and distance to monotonicity (\textsf{DTM}) are special cases of \textsf{LCS} and \textsf{ED} that are also studied in the streaming model~\cite{DBLP:conf/soda/GopalanJKK07, saks2013space}. In these two problems, one of the strings is a permutation of numbers in $[n]$ and the second string is the sorted permutation $\langle 1,2,\ldots,n\rangle$. Therefore, for these special cases $\of[i]$ is always equal to $i$. As a result, our algorithms for \textsf{ED} and \textsf{LCS} can be seen as a generalization of previous works on streaming \textsf{LIS} and distance to monotonicity.


\input{table}

\subsection{Related work}\label{sec:related}
Quadratic time solutions for  \textsf{ED}\ and \textsf{LCS}\ have
been known for many decades~\cite{leiserson2001introduction}. Recently, it has been shown that a truly subquadratic time
solution for either \textsf{ED} or \textsf{LCS} refutes {\em Strong
Exponential Time Hypothesis} (\textsf{SETH}), a conjecture widely believed
in the community
(see~\cite{backurs2015edit,abboud2015tight,lcs-hardness3}). Therefore, much attention is given to
approximation algorithms for the
two problems. For edit distance, a series of
works~\cite{landau1998incremental},~\cite{bar2004approximating},~\cite{batu2006oblivious},
and~\cite{andoni2012approximating} improve the approximation factor
culminating in the seminal work of Andoni, Krauthgamer, and
Onak~\cite{andoni2010polylogarithmic} that finally obtains a
polylogarithmic approximation factor in near-linear time. More recently constant factor approximation algorithms with truly subquadratic runtimes are obtained for edit distance (a question which was open for a few decades): first a quantum algorithm~\cite{boroujeni2018approximating}, then a classic solution~\cite{bestpaper}, and finally for far strings, near linear time solutions are also given~\cite{im1,im2}.
\textsf{LCS} has also received tremendous  attention in recent
years~\cite{lcsnew,rubinstein2019reducing,saeedfocs19,ab17,
abboud-lcs-hardness2, cglrr18}. Only trivial solutions were known
for \textsf{LCS} until very recently: a 2 approximate solution when
the alphabet is 0/1 and an $O(\sqrt{n})$ approximate solution for
general alphabets in linear time. Both these bounds are recently
improved by Hajiaghayi \textit{et
al.}~\cite{lcsnew} and Rubinstein and
Song~\cite{rubinstein2019reducing} (see also a recent approximation
algorithms given by Rubinstein \textit{et al.}~\cite{saeedfocs19}).

Streaming algorithms for edit distance have been
limited to the case that the distance between the two strings is 
{\em bounded} by a parameter $k$ which is substantially smaller than $n$. A parameterized
streaming algorithm  that makes one-pass over its input $\on$ and $\of$
with space $O(k^6)$ (which can be as large as the input size) and
running time $O(n+k^6)$~\cite{CGK16} (STOC'16)  is presented recently as well. 
\paragraph{Independent work.} Our $\tilde O_\epsilon(\sqrt{n})$ result for $\ED$ is also achieved independently in a recent work by Cheng \textit{et al.}~\cite{cheng2020space}. However, they do not give our main result which is a constant approximation streaming algorithm for $\ED$ with the memory of $\tilde O(n^\delta)$. They also give an algorithm for finding $1+\epsilon$ approximation of $\ED$ with the memory of $O(n^\delta)$. However, their algorithm works only when we have random access to both strings, and their algorithm does not work in the streaming or asymmetric streaming model.
\subsection{Preliminaries}
For a string $\on$, we use $\on[i]$ to denote the \ith{i} character in $\on$. We use $\on[i, j]$ to denote the substring of $\on$ from the \ith{i} character to the \ith{j} character. We also use $\on[i, j)$ to denote the substring of $\on$ from the \ith{i} character to \ith{(j-1)} character ($\on[i,i)$ is an empty string). 

Given two strings $\on$ and $\of$, the longest common subsequence ($\LCS$) of $\on$ and $\of$ is a string $t$ with the maximum length such that $t$ is a subsequence of both $\on$ and $\of$. In other words, $t$ can be obtained from both $\on$ and $\of$ by removing some of the characters. We use $\lcs(\on,\of)$ to denote the length of the $\LCS$ of two strings $\on$ and $\of$. The edit distance ($\ED$) between two strings $\on$ and $\of$, denoted by $\ed(\on,\of)$, is the minimum number of character insertions, deletions, and substitutions needed to transform one string to the other string.

\textbf{Asymmetric streaming model.} Throughout this paper, we assume that the input of the algorithm consists of two strings $\of$ and $\on$. We assume for simplicity and without loss of generality that the two strings have equal length $n$. We call the string $\of$ the \textit{offline} string and assume that the algorithm has random access to the characters of $\of$ by making a query. The other string $s$ arrives as a stream of characters. We call $s$ the \textit{online} string.


\input{tri}

%% file: table.tex
\begin{table}[!htbp]
	\centering
	\begin{tabular}{|l|c|c|c|}
		\hline
		problem & approximation factor & memory & reference\\
		\hline
		\textsf{ED} & $O(2^{1/\delta})$ & $\tilde{O}(n^{\delta}/\delta)$ & Theorem \ref{thm:mc} \\
		\hline
		\textsf{ED} & $1+\epsilon$ & $\tilde O_\epsilon(\sqrt{n})$ & Theorem \ref{thm:mc2}\\
		\hline
		\textsf{LCS} & $1-\epsilon$ & $\tilde{O_\epsilon}(\sqrt{n})$ & Theorem \ref{thm:mc3}\\
		\hline
		\textsf{LIS} & $1-\epsilon$ & $\tilde{O_\epsilon}(\sqrt{n})$ & \cite{DBLP:conf/soda/GopalanJKK07} \\
		\hline
				\textsf{$n\ -$~\textsf{LCS}} & $1+\epsilon$ & $\tilde{O_\epsilon}(\sqrt{n})$ & \cite{ saks2013space} \\
								\hline
				\textsf{DTM} & $1+\epsilon$ & ${O_\epsilon}(\log^2 n)$ & \cite{ saks2013space, naumovitz2014polylogarithmic}\\
		\hline							
		\textsf{DTM} & $1+\epsilon$ & $\tilde{O_\epsilon}(\sqrt{n})$ & \cite{DBLP:conf/soda/GopalanJKK07}\\
		\hline
				\textsf{DTM} & $2$ & $O(\log^2 n)$ & \cite{DBLP:conf/soda/ErgunJ08}\\
		\hline
		\textsf{DTM} & $4$ & $O(\log^2 n)$ & \cite{DBLP:conf/soda/GopalanJKK07}\\
		\hline
	\end{tabular}
	\caption{The results of this paper along with previous work.}
\end{table}

%% file: tri.tex
\subsection{Our Technique: Triangle Inequality}
As mentioned earlier, our main result is an algorithm with memory $\tilde{O}(n^{\delta})$ for any constant $\delta > 0$ that approximates edit distance within a constant factor in the asymettric model. When the available memory is limited, a typical approach to approximating edit distance is to break each of the strings into smaller pieces and find a solution in which each piece of a string is entirely transformed into another piece of the other string. Such solutions have been referred to as ``window-compatible solutions"~\cite{boroujeni2018approximating} or ``matching between candidate intervals" in previous work~\cite{hss19} (a similar techniques is also used in~\cite{bestpaper} to obtain a constant-factor approximate solution for \textsf{ED}). One should construct the pieces in a way that there always exists such a solution whose approximation factor is bounded. Previous work give several constructions with small approximation factors~\cite{boroujeni2018approximating,bestpaper,hss19,saeedfocs19}.

Let us refer to these pieces as windows and to such solutions as window-compatible solutions. It is not hard to see that if the edit distance between every pair of windows is available, then one can find an optimal window-compatible solution without any knowledge of the strings. That is, just knowing the distances between the windows suffices to find the optimal window-compatible solution. On the other hand, computing the edit distance between each pair of windows requires memory proportional to the window sizes. Therefore, a convenient way to design a memory-efficient algorithm (in certain settings such as MPC) is to give a construction for the windows in which the maximum window size is small and that it guarantees the existence of an almost optimal window-compatible solution. 

The problem becomes more challenging in the streaming setting as the online string ($\on$) is only available in a single pass. Therefore, when the characters of a window of $\on$ are stored in the memory, we have to use that information immediately to compute the edit distance of that particular window with all windows of the offline string. If the maximum window size is $l$, then we need memory $\Omega(l)$ for that purpose. Moreover, the number of windows for such a construction should be at least $\Omega(n/l)$, otherwise some parts of the strings are not included in any window and such a construction cannot guarantee any approximation factor. Thus, one needs to keep track of $O(n/l)$  values for each window of the online string, determining its distance from the windows of the offline string. Roughly speaking, this suggests that this approach can only take us as far as obtaining a solution with memory $O(\sqrt{n})$. We more formally show in Section \ref{sec:rootn} that this technique leads to a solution with approximation factor $1+\epsilon$ and memory $\tilde{O_\epsilon}(\sqrt{n})$.

Triangle inequality is the key to improving the memory of the algorithm. The key idea is summarized in the following: consider a window $w$ of the online string for which we would like to store its distance from all windows of the offline string. Instead of directly storing these values, we find a substring $[\ell,r]$ of the offline string whose edit distance is the smallest to $w$. Let the distance be $d$. We only keep 3 integer numbers  $\ell$, $r$, $d$ for this window. Surprisingly, these 3 numbers suffice to recover a 3-approximate solution for the edit distance of $w$ from any substring of the offline strings (including all the windows) without even knowing $w$! More precisely, whenever the distance of $w$ from an interval $\of[\ell',r']$ of the offline string is desired, we approximate $\mathsf{ed}(w,\of[\ell',r'])$ by $d+\mathsf{ed}(\of[\ell,r],\of[\ell',r'])$. It is not hard to see by triangle inequality that $d+\mathsf{ed}(\of[\ell,r],\of[\ell',r'])$ is at least as large and at most 3 times larger than the actual distance between $\of[\ell',r']$ and $w$. Moreover, both substrings $\of[\ell',r']$ and $\of[\ell,r]$ are available via queries since they both belong to the offline string. Finally, when two windows of the offline string are available via queries, we show using Savitch's theorem \cite{savitch1970relationships} that their edit distance can be computed with poly-logarithmic memory.

To improve the memory of the algorithm down to $O(n^\delta)$ for any $\delta > 0$, we recursively apply the above idea to make the window sizes smaller in every recursion. This comes at the expense of a multiplicative factor of roughly $3$ in the approximation for each level of recursion. Finding the optimal window-compatible solution for our setting is also cumbersome due to memory constraints. Instead of determining that with dynamic programming, we use a brute force. This takes a significant hit on the runtime of the algorithm while keeping the memory small. More details about this algorithm is given in Section \ref{sec:constant}.

\begin{restatable}{theorem}{mainc}
\label{thm:mc}
Given an offline and online strings of length $n$ and any constant $\delta>0$, there exists a single-pass deterministic streaming algorithm that finds a $O(2^{1/\delta})$ approximation of the edit distance using $\tilde{O}(n^\delta/\delta)$ memory.
\end{restatable}

%% file: ed-constant.tex
Our main results is a streaming algorithm that given any constant $\delta>0$ finds a constant approximation of the edit distance using $\tilde{O}(n^\delta)$ memory. 
As we discussed in the previous section, instead of directly solving the edit distance, we aim to find a substring of $\on$ such that its edit distance is smallest to $\of$. We formally define this problem as follows. 

\vspace{3mm}

\begin{tbox}
\texttt{Closest Substring}

\vspace{1.5mm}

\textbf{Input:} An offline string $\of$ and an online string $\on$.

\vspace{1.5mm}

\textbf{Output:} Indices $l$, $r$ and $\ed(\of[l,r], \on)$ such that $\ed(\of[l, r],\on) \le \ed(\of[i, j],\on)$ for every $1 \le i \le j \le n$.
\end{tbox}

\vspace{3mm}

We first show that how solving the closest substring problem can give us a good approximation of the edit distance. Let $\of[l, r]$ be the substring of $\of$ with the minimum edit distance to $\on$. We know by the definition of edit distance that it satisfies the \textit{triangle inequality}\footnote{$\ed(s_1, s_3) \le \ed(s_1, s_2)+\ed(s_2, s_3)$ for any strings $s_1, s_2, s_3$.}. Therefore, we have
\begin{align}
\label{eq:triandown}
\ed(\of , \on) \le \ed(\of, \of[l, r]) + \ed(\of[l, r], \on) \,.
\end{align}
We also have,
\begin{align}
\ed&(\of, \of[l, r]) + \ed(\of[l, r], \on) \nonumber \\
&\le  \ed(\of, \on) + \ed(\of[l, r], \on) + \ed(\of[l, r], \on) &\text{By the triangle inequality.} \nonumber \\
\label{eq:trianup} 
&\le 3\ed(\of, \on) & \text{Since $\of[l,r]$ has the minimum ED to $\on$.} 
\end{align}

It follows from (\ref{eq:triandown}) and (\ref{eq:trianup}) that $\ed(\of, \of[l, r]) + \ed(\of[l, r], \on)$ is a $3$-approximation of the edit distance between $\of$ and $\on$. Therefore, if we design a streaming algorithm that finds $\of[l, r]$ and its edit distance from $\on$, we can then estimate the edit distance of $\on$ and $\of$ by computing $\ed(\of, \of[l, r]) + \ed(\of[l, r], \on)$. In the following theorem we show that $\ed(\of, \of[l, r]) + \ed(\of[l, r], \on)$ can be computed using a poly-logarithmic memory. In specific, we show that the edit distance between any two substrings of the offline string can be computed using a very small memory of $O(\log^2 n)$. The proof is available in Appendix \ref{sec:appx}.

\begin{restatable}{theorem}{logtwo}
\label{thm:logtwo}
Suppose that we have random access to two given strings $\on$ and $\of$ of length $n$. Then $\lcs(\on, \of)$ and $\ed(\on, \of)$ can be computed using $O(\log^2 n)$ memory.
\end{restatable}

Therefore, by finding the substring that has the minimum edit distance to $\on$, we can get a good approximation of the edit distance. Nonetheless, we do not know any  streaming algorithm with the memory of $O(n^\delta)$ for finding closest substring, and our algorithm only finds an approximate solution for this problem. In other words, it finds a substring of $\of$ such that its approximate edit distance to $\on$ is close to the minimum. In the rest of the section, we show that how we can approximately solve the closest substring problem with the memory of $\tilde{O}(n^\delta)$. Given an online string, we divide the online string into $n^{1-\delta}$  windows of size $n^\delta$. Our algorithm (formally as Algorithm \ref{alg:nepsilon}), then recursively finds substrings of $\of$ that have the minimum edit distance from each of these windows. Note that for each window we can store the result of solving the closest substring problem in $O(\log n)$ (We can store only three numbers which are the start and the end of the interval and the approximate edit distance to the online string). Therefore, by the end of all recursive calls our algorithm needs to store $O(n^\delta)$ values.

\begin{algorithm} [h]
 \KwData{An offline string $\of$ of length $n$, a stream of characters of the online string $\on$, and a parameter $\delta>0$.}
 \begin{algorithmic} [1]
 \IF { $|\on| \le n^\delta$}
 \STATE Store all characters of $\on$ in the memory.
 \STATE Find a substring of $\of$ that has the minimum edit distance to $\on$. Let $\of[l, r]$ be this substring and $d$ be its edit distance.
 \RETURN $l$, $r$ and $d$.
 \ELSE
 \STATE $\xi \leftarrow n^\delta$.
 \STATE Divide $\on$ into $\xi$ windows $\on^*_1, \on^*_2, \ldots, \on^*_\xi$ of size $|\on|/\xi$.
 \FOR { $i \in [\xi]$}
 \STATE Recursively find the closest substring of $\of$ from $\on^*_i$  . Let $l_i, r_i$ be the start and the end of this substring respectively, and $d_i$ be the approximate edit distance of this substring to $\on^*_i$.
 \ENDFOR
 \STATE $min\_dist \leftarrow \infty$. 
 \FOR {$ 1\le p_0\le p_1 \le \ldots \le p_{\xi} \le n+1$}
 
 \STATE $dist=\sum_{i=1}^{\xi} d_i+ \ed\big(\of[p_{i-1}, p_i), \of[l_i, r_i]\big)$.
 \IF { $dist < min\_dist$}
	\STATE $min\_dist \leftarrow dist$.
	\STATE $l \leftarrow p_0$.
	\STATE $r \leftarrow p_\xi -1$. 
 \ENDIF 
 \ENDFOR
 \RETURN $l$, $r$ and $min\_dist$.
 \ENDIF 
 \end{algorithmic}
\caption{Algorithm $\neps$ for approximating ED.}
\label{alg:nepsilon}

\end{algorithm}

In order to find the solution of the closest substring problem using these partial solutions, our algorithm considers all different substrings $\of[l, r]$ of $\of$ and all different mappings between the windows of the $\on$ and the substrings of $\of[l,r]$. Then, for any mapping it estimates the edit distance between a window of $\on$ and its mapped substring of $\of[l, r]$ using the solution of the closest substring problem that we have found in the recursive call.

In order to analyze our algorithm, we first show that finding any approximation of the closest substring problem, can yield us an approximation for the edit distance. We first define an approximate version of the closest substring problem as follows.

\begin{definition}
Given an offline string $\of$ and online string $\on$, we say that the substring $\of[l,r]$ along with its approximate edit distance $d$ is an $\alpha$-approximation for the closest substring problem if for any substring $\of[l^*,r^*]$ we  have
\begin{align}
\label{eq:cspr}
\ed(\of[l,r],\on) \le d \le \alpha \cdot \ed(\of[l^*,r^*],\on) \,. 
\end{align}
\end{definition}
In the following claim we show that we can use any $\alpha$-approximation of the closest substring problem to get a $O(\alpha)$-approximation for the edit distance.

\begin{claim}
\label{lem:closetoed}
Let $\of[l,r]$ be an $\alpha$ approximation of the closest substring problem and let $d$ be its approximate edit distance to $\on$. Then for any substring $\of[l^*,r^*]$,
  $d+\ed\big(\of[l,r],\of[l^*,r^*]\big)$
  is a $(2\alpha+1)$-approximation for the edit distance between $\of[l^*,r^*]$ and $\on$.
\end{claim}
\begin{proof}
First we show that $\big(d+\ed\big(\of[l,r],\of[l^*,r^*]\big)\big)$ is not less than the edit distance between $\of[l^*,r^*]$ and $\on$.
\begin{align*}
d+\ed\big(\of[l,r],\of[l^*,r^*]\big) &\ge \ed(\of[l,r],\on) + \ed\big(\of[l,r],\of[l^*,r^*]\big) & \text{By (\ref{eq:cspr}).} \\
&\ge \ed\big(\of[l^*,r^*],\on\big) \,. & \text{By the triangle inequality.}
\end{align*}
We now show that the value of $\big(d+\ed\big(\of[l,r],\of[l^*,r^*]\big)\big)$ is at most $(2\alpha+1) \cdot \ed\big(\on, \of[l^*,r^*]\big)$. Thus it gives us a $(2\alpha+1)$-approximation of the edit distance. We have
\begin{align*}
d+\ed\big(\of[l,r],\of[l^*,r^*]\big) &\le d+\ed\big(\on,\of[l,r]\big)+ \ed\big(\on,\of[l^*,r^*]\big) & \text{By the triangle inequality.} \\
& \le d+ \alpha \cdot \ed\big(\on,\of[l^*,r^*]\big)+ \ed\big(\on,\of[l^*,r^*]\big) & \text{By (\ref{eq:cspr}).}\\
& = d+ (\alpha+1) \cdot \ed\big(\on,\of[l^*,r^*]\big) \\
& \le \alpha \cdot \ed\big(\on,\of[l^*,r^*]\big) + (\alpha+1) \cdot \ed\big(\on,\of[l^*,r^*]\big) & \text{By (\ref{eq:cspr}).}\\
& = (2\alpha+1) \cdot \ed\big(\on,\of[l^*,r^*]\big) \,,
\end{align*}
which completes the proof of the claim.
\end{proof}

Based on our discussion above, we design an algorithm that finds a constant approximation of the edit distance  using $\tilde{O}(n^\delta)$ memory for any $\delta>0$. The algorithm first divides the online string into $n^\delta$ windows with the equal length. Therefore, the length of each window is $n^{1-\delta}$. It then finds an approximate solution of the closest substring problem for each window recursively. By Claim \ref{lem:closetoed}, we can use the approximate solution of the closest substring problem for each window, to find its edit distance from every other substring of the offline string. The algorithm uses these approximate solutions to approximate the edit distance between the entire online string and any substring of the offline string.
 
Note that by each recursive call the length of the online string will get smaller by a multiplicative factor of $n^{-\delta}$. Therefore, when the depth of the recursive calls becomes $1/\delta$, the length of the remaining online string is bounded by $O(n^\delta)$ and we can store all of this remaining online string in the memory and find the exact solution of the closest substring problem. Thus, the depth of the recursion is bounded by $O(1/\delta)$. In the following theorem we show that the approximation ratio of our algorithm is $O(2^{1/\delta})$. 
\begin{theorem}
\label{thm:alg1}
Given an offline string $\of$, an online string $\on$ and any constant $\delta>0$, let $n$ be the length of the offline string and $n^\gamma$ be the length of the online string where $\gamma>0$. Then, Algorithm \ref{alg:nepsilon} finds a $O\big(2^{\gamma/\delta}\big)$ approximation for the closest substring problem.
\end{theorem}

\begin{proof}
We use induction on the length of the online string to prove the theorem. In specific, using induction on $\gamma$ we show that the approximation ratio of the algorithm is bounded by $2^{\lceil \gamma/\delta \rceil+1}-1$.
If  the length of the online string is at most $n^\delta$, then the algorithm stores all of the characters of the online string and find the exact solution. In other words, for $\gamma \le \delta$, the algorithm finds the exact solution. Thus, its approximation ratio is $1$ and the induction clearly holds.

Otherwise, we can assume the length of the online string is $n^\gamma$ where $\gamma > \delta$. In that case the algorithm divides the online string into $n^\delta$ windows of equal length. For the simplicity of the presentation, we assume that the length of the online string is divisible by $n^\delta$. Therefore, the algorithm divides $\on$ into $n^\delta$ windows $\on^*_1, \on^*_2, \cdots, \on^*_{n^\delta}$ each with the length of $n^{\gamma- \delta}$, and we have $\on^*_i = \on[(i-1) \cdot n^{\gamma - \delta}+1, i \cdot n^{\gamma -\delta}]$. The algorithm then recursively finds the closest substring of $\of$ for each of these windows. For the window $\on^*_i$, let $\of[l_i,r_i]$ be the substring returned by the algorithm and let $d_i$ be its approximate edit distance from $\on^*_i$. By the induction hypothesis we have that the approximation ratio of the solution for each window is bounded by
\begin{align*}
2^{\lceil (\gamma-\delta)/\delta \rceil+1}-1 = 2^{\lceil \gamma/\delta \rceil}-1 \,.
\end{align*}

Let $\of[l^*,r^*]$ be an arbitrary substring of $\of$. Consider the optimal mapping between $\on^*_i$ windows and $\of[l^*,r^*]$. Let assume that in the optimal mapping, window $\on^*_i$ is mapped to $\of[p^*_{i-1},p^*_i)$ (see Figure \ref{figs:windows}) where
\begin{align*}
l^*=p^*_0\le p^*_1 \le \cdots \le p^*_{n^\delta} = r^*+1 \,.
\end{align*}
Since $p^*_0, p^*_1, \ldots, p^*_{n^\delta}$ is the optimal mapping, we have
\begin{align}
\label{eq:optimalmapping}
\ed(\on,\of[l^*,r^*]) = \sum_{i=1}^{n^\delta} \ed\big(\on^*_i,\of[p^*_{i-1},p^*_i)\big) \,.
\end{align}

\input{figs/windows}

Recall that for each window $\on^*_i$, the substring $\of[l_i,r_i]$ and the distance $d_i$ is a $\big(2^{\lceil \gamma/\delta \rceil}-1\big)$ approximation of the closest substring problem. Therefore by Claim \ref{lem:closetoed} we can use this approximate solution to estimate the edit distance between $\on^*_i$ and other substrings of $\of$. By this claim $d_i + \ed\big(\of[l_i,r_i],\of[p^*_{i-1},p^*_i)\big)$ is a $(2^{\lceil \gamma/\delta \rceil+1}-1)$-approximation for the edit distance between $\on^*_i$ and $\of[p^*_{i-1},p^*_i)$. In specific, 
\begin{align}
\label{eq:appxmapping}
d_i + \ed\big(\of[l_i,r_i],\of[p^*_{i-1},p^*_i)\big) \le \big(2^{\lceil \gamma/\delta \rceil+1}-1\big) \cdot \ed\big(\on^*_i,\of[p^*_{i-1},p^*_i)\big) \,.
\end{align}

For each substring $\of[l^*,r^*]$, Algorithm \ref{alg:nepsilon} iterates over all different mappings between $\on^*_i$ windows and this substring. Note that in order to iterate over all different mappings, we can iterate over the variables $p_0, p_1, \cdots, p_{n^\delta}$ such that
\begin{align*}
l^*=p_0\le p_1 \le \cdots \le p_{n^\delta} = r^*+1 \,,
\end{align*}
 and these variables can be stored in a memory of $\tilde{O}(n^\delta)$. For each different mapping the algorithm estimates the edit distance of each window and the mapped substring using Claim \ref{lem:closetoed}. We claim that for each substring $\of[l^*,r^*]$, the algorithm finds $(2^{\lceil \gamma/\delta \rceil+1}-1)$-approximation of the edit distance between this substring and the online string. To show that consider the optimal mapping $p^*_0, p^*_1, \cdots, p^*_{n^\delta}$, then the distance that algorithm estimates is bounded by
 \begin{align*}
 &\sum_{i=0}^{n^\delta} d_i + \ed\big(\of[l_i,r_i],\of[p^*_{i-1},p^*_i)\big)\\
 &\le \sum_{i=0}^{n^\delta} \big(2^{\lceil \gamma/\delta \rceil+1}-1\big) \cdot \ed\big(\on^*_i,\of[p^*_{i-1},p^*_i)\big) & \text{By (\ref{eq:appxmapping}).} \\
 &=  \big(2^{\lceil \gamma/\delta \rceil+1}-1\big) \sum_{i=0}^{n^\delta} \ed\big(\on^*_i,\of[p^*_{i-1},p^*_i)\big) \\
 &=  \big(2^{\lceil \gamma/\delta \rceil+1}-1\big) \ed(\on,\of[l^*,r^*]) \,. & \text{By (\ref{eq:optimalmapping}).}
 \end{align*}
 Therefore for each substring $\of[l^*,r^*]$, the algorithm finds a $\big(2^{\lceil \gamma/\delta \rceil+1}-1\big)$ approximation of its edit distance to $\on$. Thus, the algorithm finds a $\big(2^{\lceil \gamma/\delta \rceil+1}-1\big)$ approximation of the closest substring problem. This completes the induction and proves the theorem.
\end{proof}

\mainc*
\begin{proof}
By Theorem \ref{thm:alg1}, Algorithm \ref{alg:nepsilon} finds a $O(2^{1/\delta})$ approximation of the closest substring problem. Recall that by Theorem \ref{thm:logtwo}, we can find the edit distance of any two substrings of $\of$ using a very small memory. Therefore by Claim \ref{lem:closetoed}, we can find a $O(2^{1/\delta})$ approximation of the edit distance between $\on$ and $\of$.

Now we show that the memory of Algorithm \ref{alg:nepsilon} is at most $\tilde{O}(n^\delta/\delta)$. While the length of the online string is larger than $n^\delta$, Algorithm \ref{alg:nepsilon} divides the online string into $n^\delta$ windows and recursively solves the closest substring problem for each window. Therefore, by each recursive call the length of the online string will decrease by a multiplicative factor of $n^{-\delta}$. Thus, the maximum depth of the recursive calls is bounded by $O(1/\delta)$. At each call the algorithm acquires a memory of $\tilde{O}(n^\delta)$ which is the memory needed for storing the result of the recursive calls and iterating over all possible mappings. Therefore, the memory of the algorithm is bounded by $\tilde{O}(n^\delta/\delta)$.
\end{proof}

%% file: figs/windows.tex
\begin{figure}[t]
\usetikzlibrary{patterns}

\begin{center}
\begin{tikzpicture}[scale=0.86, transform shape]

\draw (0,0) -- (18,0) -- (18,1) -- (0,1) -- (0,0);

\draw (0.3,0) -- (0.3,1);
\draw  (0,0) rectangle (0.3,1);
\draw (0.6,0) -- (0.6,1);
\draw (0.9,0) -- (0.9,1);

\node[text width=1cm] at (1.8,0.4) {$\ldots$};
\node[text width=1cm] at (2.8,0.4) {$\ldots$};
\node[text width=1cm] at (3.8,0.4) {$\ldots$};

\draw (4.1,0) -- (4.1,1);
\draw (4.4,0) -- (4.4,1);
\draw (4.7,0)  -- (4.7,1);
\draw (5,0) -- (5,1);

\draw [dashed] (0.15,0.1) rectangle (4.85,0.9);
\draw [fill=gray, opacity=0.2] (0.15,0.1) rectangle (4.85,0.9);

\draw (5.3,0) -- (5.3,1);
\draw  (5,0) rectangle (5.3,1);
\draw (5.6,0) -- (5.6,1);
\draw (5.9,0) -- (5.9,1);

\node[text width=1cm] at (6.8,0.4) {$\ldots$};
\node[text width=1cm] at (7.8,0.4) {$\ldots$};
\node[text width=1cm] at (8.8,0.4) {$\ldots$};

\draw (9.1,0) -- (9.1,1);
\draw (9.4,0) -- (9.4,1);
\draw (9.7,0)  -- (9.7,1);
\draw (10,0) -- (10,1);

\draw [dashed] (5.15,0.1) rectangle (9.85,0.9);
\draw [fill=gray, opacity=0.2] (5.15,0.1) rectangle (9.85,0.9);

\draw (13,0) -- (13,1);
\node[text width=1cm] at (11.5,0.4) {$\ldots$};

\draw (13.3,0) -- (13.3,1);
\draw  (13,0) rectangle (13.3,1);
\draw (13.6,0) -- (13.6,1);
\draw (13.9,0) -- (13.9,1);

\node[text width=1cm] at (14.8,0.4) {$\ldots$};
\node[text width=1cm] at (15.8,0.4) {$\ldots$};
\node[text width=1cm] at (16.8,0.4) {$\ldots$};

\draw (17.1,0) -- (17.1,1);
\draw (17.4,0) -- (17.4,1);
\draw (17.7,0)  -- (17.7,1);
\draw (18,0) -- (18,1);

\draw [fill=gray, opacity=0.2] (13.15,0.1) rectangle (17.85,0.9);
\draw [dashed] (13.15,0.1) rectangle (17.85,0.9);

\node[text width=0.3cm] at (0.15,1.3) 
{   \huge $\uparrow$ };
	\node[text width=0.3cm] at (0.15,1.9) 
{   \large $1$ };

\node[text width=0.3cm] at (5.15,1.3) 
{   \huge $\uparrow$ };
\node[text width=2cm] at (5.55,1.9) 
{   \large $n^{\gamma- \delta} +1$ };

\node[text width=0.3cm] at (10.05,1.3) 
{   \huge $\uparrow$ };
\node[text width=2cm] at (10.45,1.9) 
{   \large $2n^{\gamma- \delta} +1$ };

\node[text width=0.3cm] at (13.15,1.3) 
{   \huge $\uparrow$ };
\node[text width=2cm] at (13.55,1.9) 
{   \large $n- n^{\gamma- \delta} +1 $ };

\node[text width=0.3cm] at (17.85,1.3) 
{   \huge $\uparrow$ };
\node[text width=0.3cm] at (17.90,1.9) 
{   \large $n$ };

\node[text width=0.5cm] at (-0.40,0.4) 
{   \huge $\on$ };

\draw (0,-1) -- (18,-1) -- (18,-2) -- (0,-2) -- (0,-1);

\draw (0.3,-1) -- (0.3,-2);
\draw  (0,-1) rectangle (0.3,-2);
\draw (0.6,-1) -- (0.6,-2);
\draw (0.9,-1) -- (0.9,-2);

\node[text width=1cm] at (1.8,-1.6) {$\ldots$};

\draw (2.1,-1) -- (2.1,-2);
\draw (2.4,-1) -- (2.4,-2);
\draw (2.7,-1)  -- (2.7,-2);
\draw (3,-1) -- (3,-2);

\node[text width=1cm] at (4.15,-1.6) {$\ldots$};


\draw (4.8,-1) -- (4.8,-2);
\draw (5.1,-1) -- (5.1,-2);
\draw (5.4,-1)  -- (5.4,-2);
\draw (5.7,-1) -- (5.7,-2);

\node[text width=1cm] at (7.55,-1.6) {$\ldots$};


\draw (9.1,-1) -- (9.1,-2);
\draw (9.4,-1) -- (9.4,-2);
\draw (9.7,-1)  -- (9.7,-2);
\draw (10,-1) -- (10,-2);

\node[text width=1cm] at (11.5,-1.6) {$\ldots$};

\draw (13,-1) -- (13,-2);
\draw (13.3,-1) -- (13.3,-2);
\draw (13.6,-1) -- (13.6,-2);
\draw (13.9,-1) -- (13.9,-2);

\node[text width=1cm] at (14.85,-1.6) {$\ldots$};

\draw (15.3,-1) -- (15.3,-2);
\draw (15.6,-1) -- (15.6,-2);
\draw (15.9,-1) -- (15.9,-2);
\draw (16.2,-1) -- (16.2,-2);

\node[text width=1cm] at (16.9,-1.6) {$\ldots$};

\draw (17.1,-1) -- (17.1,-2);
\draw (17.4,-1) -- (17.4,-2);
\draw (17.7,-1)  -- (17.7,-2);
\draw (18,-1) -- (18,-2);

\node[text width=0.3cm] at (2.55,-2.4) 
{   \huge $\downarrow$ };
\node[text width=1.5cm] at (2.7,-3) 
{   \large $p_0 = l$ };

\node[text width=0.3cm] at (5.55,-2.4) 
{   \huge $\downarrow$ };
\node[text width=0.3cm] at (5.55,-3) 
{   \large $p_1$ };

\node[text width=0.3cm] at (5.55,-2.4) 
{   \huge $\downarrow$ };
\node[text width=0.3cm] at (5.55,-3) 
{   \large $p_1$ };

\node[text width=0.3cm] at (9.85,-2.4) 
{   \huge $\downarrow$ };
\node[text width=0.3cm] at (9.85,-3) 
{   \large $p_2$ };

\node[text width=0.3cm] at (13.45,-2.4) 
{   \huge $\downarrow$ };
\node[text width=0.3cm] at (13.25,-3) 
{   \large $p_{n^\delta-1}$ };

\node[text width=0.3cm] at (15.75,-2.4) 
{   \huge $\downarrow$ };
\node[text width=0.3cm] at (15.75,-3) 
{   \large $r$ };

\node[text width=0.5cm] at (-0.40,-1.6) 
{   \huge $\of$ };

\draw [pattern=fivepointed stars] (0,-0.1) -- (4.90,-0.1) -- (5.3,-0.9) -- (2.5,-0.9) -- (0,-0.1);
\draw [pattern=bricks] (5.1,-0.1) -- (9.90,-0.1) -- (9.6,-0.9) -- (5.5,-0.9) -- (5.1,-0.1);
\draw [pattern=vertical lines] (13,-0.1) -- (18,-0.1) -- (15.8,-0.9) -- (13.4,-0.9) -- (13,-0.1);

\end{tikzpicture}
\end{center}
\caption{An optimal transformations of the windows of $\on$ into intervals of $\of[l,r]$ is shown in this figure. Gray rectangles illustrate the windows of $s$ and each pattern shows how its corresponding block is transformed into an interval of $\of[l,r]$.}\label{figs:windows}
\end{figure}
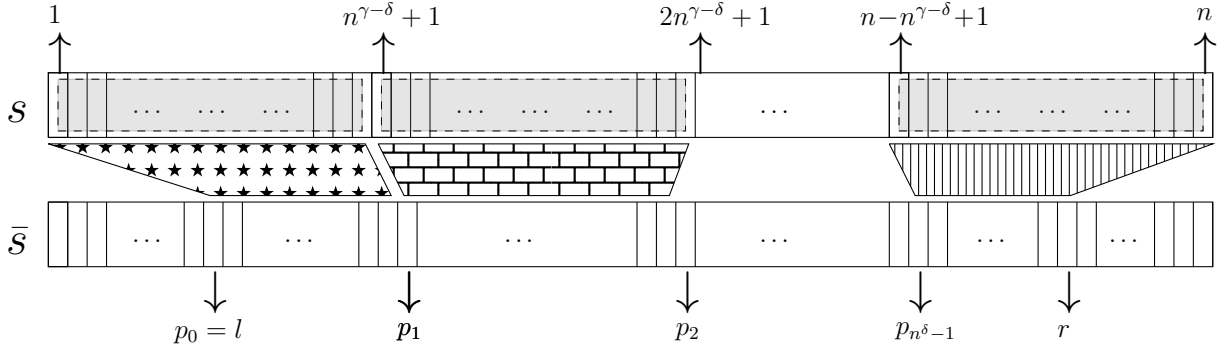

%% file: lcs-rootn.tex
In this section, we design a streaming algorithm for finding a $(1-\epsilon)$ approximation of the $\LCS$ using $\tilde{O}(\sqrt{n}/\epsilon)$ memory. We first define the $\lcsp$ function as below.

\vspace{3mm}

\begin{tbox}
\texttt{$\lcsp_{l,r}$}

\vspace{1.5mm}

\textbf{Input:}  A position $p$ in $\of$ and a non-negative integer $k$.

\vspace{1.5mm}

\textbf{Output:} The smallest position $q$ such that $\lcs(\of[p,q],\on[l,r]) \ge k$. If no such $q$ exists, the output is $\infty$. 
\end{tbox}

\vspace{3mm}

For a position $p$ in $\of$, a substring $\on[l,r]$ of $\on$, and a non-negative integer $k$, we use $\lcsp_{l,r}(p,k)$ to denote the result of the mentioned function which is the smallest position $q$ such that $\LCS$ of $\of[p,q]$ and $\on[l,r]$ is at least $k$.\footnote{We also define $\lcsp_{l,r}(p,0)$ to be $p-1$.}
It is easy to verify that the $\LCS$ of two strings $\of$ and $\on$ is equal to the largest $k$ such that $\lcsp_{1,n}(1,k) < \infty$. Therefore, instead of solving the $\LCS$ problem, we can solve the $\lcsp_{1,n}$ problem and report the largest $k$ such that $\lcsp_{1,n} (1,k) < \infty$. We start designing our algorithm, by observing some properties of the function $\lcsp$.
\begin{mobsv}
Function $\lcsp_{l,r}$ is non-decreasing on $p$ and $k$. In other words, for every numbers $p_1 \le p_2$ and $k_1 \le k_2$, we have
\begin{align*}
\lcsp_{l,r}(p_1,k_1) \le \lcsp_{l,r}(p_2,k_2) \,.
\end{align*}
\end{mobsv}
\begin{proof}
It immediately derives from the definition of the function.
\end{proof}

 Consider the function $\lcsp_{l,r}(p,k)$, and let $\on[l,m]$ and $\on[m+1, r]$ be an arbitrary division of the substring $\on[l,r]$ into two substrings. The following claim shows  how we can compute $\lcsp_{l,r}$ from $\lcsp_{l,m}$ and $\lcsp{m+1,r}$.
\begin{claim}
\label{claim:lcspdiv}
For any $k \ge 0$, the following holds.
\begin{align} 
&\lcsp_{l,r}(p,k) = \nonumber \\
		&\min_{\substack{k_1,k_2\ge 0, k_1+k_2=k, \\ \lcsp_{l,m}(p, k_1)< \infty}} \quad \lcsp_{m+1,r}(\lcsp_{l,m}(p,k_1)+1, k_2) \,.
\end{align}
\end{claim}
\begin{proof}
For any $k_1, k_2 \ge 0$ such that $k=k_1+k_2$ and $\lcsp_{l,m}(p, k_1)< \infty$, the value of $\lcsp_{m+1,r}(\lcsp_{l,m}(p,k_1)+1, k_2)$ indicates the ending of a common subsequence of size $k$ such that exactly $k_1$ characters from $\on[l,m]$ are in this common subsequence and $k_2$ characters from $\on[m+1,r]$ are in this subsequence. Therefore, we always have
\begin{align*} 
&\lcsp_{l,r}(p,k) \le \nonumber \\
		&\min_{\substack{k_1,k_2\ge 0, k_1+k_2=k, \\ \lcsp_{l,m}(p, k_1)< \infty}} \quad \lcsp_{m+1,r}(\lcsp_{l,m}(p,k_1)+1, k_2) \,.
\end{align*}
In order to complete the proof of the claim, we show that there always exists $k_1$ and $k_2$ such that $k_1+k_2 =k$ and $\lcsp_{m+1,r}(\lcsp_{l,m}(p,k_1)+1, k_2) \le \lcsp_{l,r}(p,k)$.
 
Consider an optimal solution of the function $\lcsp_{l,r}(p,k)$, and let suppose that $q=\lcsp_{l,r}(p,k)$. In this solution there exists a common subsequence of size $k$ between the characters in $\on[l,r]$ and $\of[p,q]$. Let suppose that in that solution character $\on[a_i]$ is matched to $\of[b_i]$ for each $1 \le i \le k$. W.l.o.g., we can assume
\begin{align*}
l \le a_1 < a_2 < \cdots < a_k \le r \,. 
\end{align*}
It also implies that
\begin{align*}
p \le b_1 < b_2 < \cdots < b_k=q \,.
\end{align*}
We consider two different cases. The first case is when all indices $a_i$ are larger than $m$. In this case $k$ characters from $\on[m+1,r]$ are matched to $\of[p,q]$. Therefore, $\lcsp_{m+1,r}(p,k) \le q$. By setting $k_1=0$ and $k_2=k$, we get
\begin{align*}
\lcsp_{m+1,r}(\lcsp_{l,m}(p,k_1)+1, k_2) &= \lcsp_{m+1,r}((p-1)+1, k_2) \\
&= \lcsp_{m+1,r}(p, k) \\
&\le q \,. 
\end{align*}

The other case is when for at least one $a_i$, we have $a_i \le m$. Let assume that $k_1$ the largest number such that $a_{k_1}$ is at most $m$.
Then, in optimal solution of $\lcsp_{l,r}(p,k)$ exactly $k_1$ characters from $\on[l,m]$ are matched to the characters in $\of[p,b_{k_1}]$. 
Therefore, we have
\begin{align}
\label{eq:lcspd1}
\lcsp_{l,m}(p, k_1) \le b_{k_1} \,.
\end{align}

We also know that there are $k_2=k-k_1$ characters from $\on[m+1,r]$ that are matched to the characters in $\of[b_{k_1}+1,q]$. Therefore we have

\begin{align}
\label{lcspd2}
\lcsp_{m+1,r}(b_{k_1}+1, k_2) \le q \,.
\end{align}

Thus, 
\begin{align*}
&\lcsp_{m+1,r}(\lcsp_{l,m}(p,k_1)+1, k_2) \\
& \le \lcsp_{m+1,r}(b_{k_1}+1, k_2) & \text{By (\ref{eq:lcspd1}).} \\
& \le q & \text{By (\ref{lcspd2}).} \\
& = \lcsp_{l,r}(p,k) \,,
\end{align*}
which proves the claim.
\end{proof}

\begin{algorithm} [h]
 \KwData{An offline string $\of$ of length $n$, a stream of characters of the online string $\on$, and an $\epsilon^*>0$.}
 \begin{algorithmic} [1]
 \STATE Divide $\on$ into $\sqrt{n}$ windows $\on^*_1, \on^*_2, \ldots, \on^*_{\sqrt{n}}$ of size $\sqrt{n}$.
 \STATE $D \leftarrow$ an array of size $\lfloor\log_{1+\epsilon^*} n \rfloor$ initially containing $\infty$ in all cells.
 \FOR { $i \in \big[\sqrt{n}\big]$}
 	\STATE $T \leftarrow$ an array of size $\lfloor\log_{1+\epsilon^*} n \rfloor$ initially containing $\infty$ in all cells.
 	\FOR {$0 \le k \le \lfloor\log_{1+\epsilon^*} n \rfloor$}
 		\STATE $T[k] \leftarrow \lcsp_{(i-1) \sqrt{n}+1, i \sqrt{n}} \big(1,\lfloor (1+\epsilon^*)^k \rfloor\big)$.
 		\FOR {$0 \le k_1 \le k$}
 			\IF {$D[k_1] < \infty$}
 				\STATE Find $\lcsp_{(i-1) \sqrt{n}+1, i \sqrt{n}} \big(D[k_1]+1, \lfloor (1+\epsilon^*)^k \rfloor- \lfloor(1+\epsilon^*)^{k_1}\rfloor \big)$ using any offline algorithm. Let $q$ be this result.
 				\STATE $T[k] \leftarrow \min\big\{T[k], q\big\}$.
 			\ENDIF		
 		\ENDFOR 
 	\ENDFOR
 	\STATE $D \leftarrow T$.
 \ENDFOR
 \RETURN The largest value $\lfloor (1+\epsilon^*)^k \rfloor$ such that $D[k] < \infty$. 
 \STATE \quad \textbf{return} $0$ if no such $k$ exists.
 \end{algorithmic}
\caption{Algorithm $\lcsrootn$ for approximating the $\LCS$.}
\label{alg:lcsrootn}
\end{algorithm}

Algorithm \ref{alg:lcsrootn} first divides the online string into $\sqrt{n}$ windows of equal sizes. We assume w.l.o.g., that length of the strings is divisible by $\sqrt{n}$. Otherwise we can always pad offline and online strings with different characters that are not in $\Sigma$ such that their new length get divisible by $\sqrt{n}$. 
The algorithm divides $\on$ into $\sqrt{n}$ windows $\on^*_1,\on^*_2, \cdots, \on^*_{\sqrt{n}}$ each with the size of $\sqrt{n}$ where $\on^*_i$ is the substring $\on[(i-1)\sqrt{n}+1, i\sqrt{n}]$. Given an $\epsilon^*>0$, the algorithm keeps an array $D$ of the size $\lfloor \log_{1+\epsilon^*} n \rfloor$ where $D[k]$ is an estimation of $\lcsp(1,\lfloor (1+\epsilon^*)^k \rfloor)$ in the subsequence of the online string that has arrived so far in the stream. Specifically, after arrival of the window $\on^*_i$ in the stream, the algorithm keeps an estimation of $\lcsp_{1,i \sqrt{n}} (1, \lfloor (1+\epsilon^*)^k \rfloor)$ in $D[k]$. First we show that how the algorithm can update the array $D$ upon arrival of a new window, and after that we demonstrate the approximation guarantee of our method.

Let assume that we have an array $D$ in which $D[k]$ is an approximation of $\lcsp_{1, (i-1) \sqrt{n}} \big(1, \lfloor (1+\epsilon^*)^k\rfloor\big)$ for different values of $0 \le k \le \lfloor \log_{1+\epsilon^*} n \rfloor$. Upon arrival of a new window $\on^*_i$, the algorithm has to update the array $D$. Suppose that we want to find $\lcsp_{1, i \sqrt{n}} \big(1, \lfloor(1+\epsilon^*)^k\rfloor \big)$. According to Claim \ref{claim:lcspdiv}, there are integers $k_1^*$ and $k_2^*$ such that $k_1^*+k_2^* = \lfloor(1+\epsilon^*)^k \rfloor$ and
\begin{align}
\label{eq:nextblock}
\lcsp_{1,i \sqrt{n}} (1,\lfloor(1+\epsilon^*)^k \rfloor) = \lcsp_{(i-1) \sqrt{n}+1, i \sqrt{n}} (\lcsp_{1,(i-1)\sqrt{n}}(1,k^*_1)+1,k^*_2) \,.
\end{align}

The algorithm stores all of the characters of $\on^*_i$ in the memory. Therefore, for every $p$ and $k$ we can compute the function $\lcsp_{(i-1)\sqrt{n}+1,i \sqrt{n}} (p,k)$ using only poly-logarithmic extra memory (see Theorem \ref{thm:logtwo}). In order to update the array $D$, the algorithm iterates over all $k^*_1$ such that $k^*_1$ is a power of $(1+\epsilon^*)$ and pick the one that minimizes the r.h.s. of (\ref{eq:nextblock}). Specifically, let $T$ an array of length $\lfloor \log_{1+\epsilon^*} n \rfloor$ which represents the updated estimates after arrival of $\on^*_i$. Initially for each $k$ we set
\begin{align*}
T[k]=\lcsp_{(i-1)\sqrt{n}+1,i \sqrt{n}} \big(1,\lfloor(1+\epsilon^*)^k\rfloor\big) \,,
\end{align*}
which represents the case that all characters in the optimal solution of $\lcsp_{1, i \sqrt{n}} \big(1, \lfloor(1+\epsilon^*)\rfloor^k\big)$ are from the window $\on^*_i$, i.e., when $k^*_1$ is zero in (\ref{eq:nextblock}). Then the algorithm considers values of $k^*_1$ such that $k^*_1$ is a power of $(1+\epsilon^*)$, i.e., we have $k^*_1 = \lfloor(1+\epsilon^*)^{k_1} \rfloor$ for some integer $k_1$. Recall that our algorithm makes sure that that $D[k]$ is an approximation of $\lcsp_{1, (i-1) \sqrt{n}} \big(1, \lfloor (1+\epsilon^*)^k \rfloor \big)$. Therefore we can approximate the r.h.s. of (\ref{eq:nextblock}) for $k^*_1 = \lfloor(1+\epsilon^*)^{k_1} \rfloor$ by computing
\begin{align*}
\lcsp_{(i-1) \sqrt{n}+1, i \sqrt{n}} (D[k_1]+1,k^*_2) \,,
\end{align*}
where $k^*_2= \lfloor (1+\epsilon^*)^k \rfloor - k^*_1$. In our algorithm we compute the value above for all different value of $k^*_1$ and set the $T[k]$ equal to minimum of these values. In other words, by the end of the arrival of the window $\on^*_i$, we have
\begin{align}
\label{eq:tk}
T[k]= \min\left\{\begin{array}{lr}
        \lcsp_{(i-1)\sqrt{n}+1,i \sqrt{n}} \big(1,\lfloor (1+\epsilon^*)^k \rfloor \big) \,, \\
        \min_{\substack{k^*_1,k^*_2\ge 0, k^*_1+k^*_2=\lfloor (1+\epsilon^*)^k \rfloor , \\ k^*_1=\lfloor (1+\epsilon^*)^{k_1} \rfloor, \\ D[k_1]< \infty}} \quad \lcsp_{(i-1)\sqrt{n}+1,i \sqrt{n}}(D[k_1]+1, k^*_2)
        \end{array}\right\} \,.
\end{align}
After computing the values in the array $T$, we can replace values in the array $D$ with the values in $T$, and update the array $D$.

In order to provide an approximation guarantee for our algorithm, we first prove the following claim.
\begin{claim}
\label{claim:di}
Let $D_i$ be the array $D$ after arrival of the window $\on^*_i$, then for each $1 \le k^* \le n$, there exists a $0 \le k \le \lfloor \log_{1+\epsilon^*} n \rfloor$ such that
\begin{align*}
k^*(1-\epsilon^*)^i \le \lfloor(1+\epsilon^*)^k \rfloor \le k^* \,, 
\end{align*}
and,
\begin{align*}
D_i[k] \le \lcsp_{1, i \sqrt{n}} (1, k^*) \,.
\end{align*}
\end{claim}
\begin{proof}
We prove the claim by induction on $i$ which represents the number of windows that have arrived in the stream. For $i=1$, the algorithm finds the exact solution of $\lcsp_{1,\sqrt{n}} (1, \lfloor(1+\epsilon^*)^k \rfloor)$ for all $0 \le k \le \lfloor \log_{1+\epsilon^*} n \rfloor$. Consider an integer $1 \le k^* \le n$, then there exists some number with the form of $\lfloor (1+\epsilon^*)^k \rfloor$ between $k^*/(1+\epsilon^*)$ and $k^*$. Let $\lfloor (1+\epsilon^*)^k\rfloor$ be that number. Then we have,
\begin{align*}
D_1[k]=\lcsp_{1,\sqrt{n}} (1, \lfloor (1+\epsilon^*)^k\rfloor) \le \lcsp_{1,\sqrt{n}} (1, k^*) \,.
\end{align*}
We also have
\begin{align*}
k^*(1-\epsilon^*) \le \frac{k^*}{1+\epsilon^*} \le \lfloor (1+\epsilon^*)^k\rfloor \le k^* \,,
\end{align*}
which proves the claim for $i=1$.

Now consider an $i>1$, and a $1 \le k^* \le n$. If $\lcsp_{1, i \sqrt{n}} (1, k^*)$ is $\infty$, then the claim clearly holds. Otherwise we can assume $\lcsp_{1, i \sqrt{n}} (1, k^*) = q$ where $q < \infty$.
 By (\ref{eq:tk}) and the way our algorithm computes the array $D_{i}$ we have
 \begin{align}
\label{eq:di}
D_i[k]= \min\left\{\begin{array}{lr}
        \lcsp_{(i-1)\sqrt{n}+1,i \sqrt{n}} \big(1,\lfloor (1+\epsilon^*)^k\rfloor\big) \,, \\
        \min_{\substack{k^*_1,k^*_2\ge 0, k^*_1+k^*_2=\lfloor(1+\epsilon^*)^k \rfloor, \\ k^*_1=\lfloor(1+\epsilon^*)^{k_1} \rfloor, \\ D_{i-1}[k_1]< \infty}} \quad \lcsp_{(i-1)\sqrt{n}+1,i \sqrt{n}}(D_{i-1}[k_1]+1, k^*_2)
        \end{array}\right\} \,.
\end{align}
  By Claim \ref{claim:lcspdiv}, there exists integers $k^*_1,k^*_2 \ge 0$ such that $k^*_1+k^*_2 = k^*$ and
\begin{align*}
\lcsp_{1,i \sqrt{n}}(1,k^*) = \lcsp_{(i-1)\sqrt{n}+1 ,i \sqrt{n}}(\lcsp_{1,(i-1)\sqrt{n}}(1,k^*_1)+1, k^*_2) \,.
\end{align*} 
Let $q_1 = \lcsp_{1,(i-1)\sqrt{n}}(1,k^*_1)$, then we have
\begin{align}
\label{eq:prvq1}
\lcsp_{1,i \sqrt{n}}(1,k^*) =  \lcsp_{(i-1)\sqrt{n}+1 ,i \sqrt{n}}(q_1+1, k^*_2) \,.
\end{align} 
We consider two different cases on $k^*_1$.
\begin{itemize}
\item The first case is when $k^*_1=0$. In that case we have $q_1=0$, and by (\ref{eq:prvq1}) we have
\begin{align}
\label{eq:allinlastblock}
\lcsp_{1,i \sqrt{n}}(1,k^*) =  \lcsp_{(i-1)\sqrt{n}+1 ,i \sqrt{n}}(1, k^*) \,.
\end{align}
Then there exists some number with the form of  $\lfloor(1+\epsilon^*)^k\rfloor$  between $k^*(1-\epsilon^*)$ and $k^*$ and by (\ref{eq:di}) we have
\begin{align*}
D_i[k] &\le \lcsp_{(i-1)\sqrt{n}+1 ,i \sqrt{n}}(1, \lfloor(1+\epsilon^*)^k \rfloor) \\
&\le \lcsp_{(i-1)\sqrt{n}+1 ,i \sqrt{n}}(1, k^*) \\
&= \lcsp_{1,i \sqrt{n}}(1,k^*) \,, &\text{By (\ref{eq:allinlastblock}).} 
\end{align*}
which proves the claim for this case.
\item The other case is when $k^*_1> 0$. In this case, in the optimal solution of $\lcsp_{1,i \sqrt{n}}(1,k^*)$, exactly $k^*_1$ characters from $\on[1,(i-1)\sqrt{n}]$ are matched to the characters in $\of[1,q_1]$. By the induction hypothesis, we know there exists some $k_1$ such that
\begin{align}
\label{eq:k1bound}
k^*_1(1- \epsilon^*)^{i-1} \le \lfloor(1+\epsilon^*)^{k_1} \rfloor \le k^*_1 \,, 
\end{align}
and
\begin{align}
\label{eq:pdi}
D_{i-1}[k_1] \le \lcsp_{1,(i-1)\sqrt{n}}(1,k^*_1) = q_1 \,.
\end{align}
Let 
$k'=\lfloor(1+\epsilon^*)^{k_1} \rfloor+k^*_2$.
 Then, we have
\begin{align*}
k' &\ge k^*(1- \epsilon^*)^{i-1} \,. & \text{By (\ref{eq:k1bound}).}
\end{align*}
Let $k$ an integer such that $\lfloor(1+\epsilon^*)^k \rfloor$ is between $k'(1-\epsilon^*)$ and $k'$. We show that $D_i[k]$ satisfies the claim conditions. From the previous equation, we have
\begin{align*}
\lfloor(1+\epsilon^*)^k \rfloor &\ge k'(1-\epsilon^*) \ge k^* (1- \epsilon^*)^{i} \,.
\end{align*}

Let $k_2 = \lfloor(1+\epsilon^*)^k \rfloor- \lfloor(1+\epsilon^*)^{k_1} \rfloor$. Then we have,
\begin{align}
\label{eq:newk2}
k_2 & = \lfloor(1+\epsilon^*)^k \rfloor- \lfloor(1+\epsilon^*)^{k_1} \rfloor\nonumber \\
& \le k' -\lfloor(1+\epsilon^*)^{k_1}\rfloor \nonumber & \text{Since $\lfloor(1+\epsilon^*)^k\rfloor \le k'$.} \\
& = k^*_2 \,.& \text{Since $ k'=\lfloor(1+\epsilon^*)^{k_1} \rfloor+k^*_2$.}
\end{align}
 By (\ref{eq:di}), we have
\begin{align*}
D_i[k] &\le \lcsp_{(i-1)\sqrt{n}+1,i \sqrt{n}}(D_{i-1}[k_1]+1, k_2) \\
& \le \lcsp_{(i-1)\sqrt{n}+1,i \sqrt{n}}(q_1+1, k_2) & \text{By (\ref{eq:pdi}).} \\
& \le \lcsp_{(i-1)\sqrt{n}+1,i \sqrt{n}}(q_1+1, k^*_2) & \text{By (\ref{eq:newk2}).} \\
& = \lcsp_{1,i \sqrt{n}}(1,k^*) \,. & \text{By (\ref{eq:prvq1}).}
\end{align*}
This proves the second case and completes the induction and proves the claim.
\end{itemize}
\end{proof}
\begin{theorem}
\label{thm:m2p}
For any $\epsilon^*>0$, Algorithm \ref{alg:lcsrootn} finds a $(1-\epsilon^*)^{\sqrt{n}}$ approximation of the $\LCS$ between $\of$ and $\on$ using $\tilde{O}(\sqrt{n}+\log_{1+\epsilon^*} n)$ memory.
\end{theorem}

\begin{proof}
Let $\opt$ be the size of $\LCS$ between $\of$ and $\on$. Then $\opt$ is the largest $k$ such that $\lcsp_{1,n} (1,k) < \infty$. Our algorithm approximately computes the function $\lcsp_{1,n}$ and return the largest $k$ such that $\lcsp_{1,n} (1,k) < \infty$. By Claim \ref{claim:di}, in the final array $D$ computed by the algorithm there exists an integer $k$ such that 
\begin{align}
\label{eq:lcsapx1}
\lfloor(1+\epsilon^*)^k\rfloor \ge \opt (1-\epsilon^*)^{\sqrt{n}} \,,
\end{align}
and
\begin{align*}
D[k] \le \lcsp_{1,n}(1,\opt) < \infty \,.
\end{align*}
Therefore, the answer returned by the algorithm is at least $\lfloor(1+\epsilon^*)^k\rfloor$. By (\ref{eq:lcsapx1}) it gives us a $(1-\epsilon^*)^{\sqrt{n}}$ approximation.

To show the memory bound of Algorithm \ref{alg:lcsrootn}, observe that the algorithm needs a memory of $\tilde{O}(\sqrt{n})$ to store each window $\on^*_i$ and compute the $\LCS$ between a substring of this window and a substring of the offline string (using Theorem \ref{thm:logtwo}). Also, the algorithm keeps an array $D$ and $T$ of size $\lfloor \log_{1+\epsilon^*} n \rfloor$. Therefore, the memory of the algorithm is bounded by $\tilde{O}(\sqrt{n}+\log_{1+\epsilon^*} n)$.
 
\end{proof}

\begin{theorem}
\label{thm:mc3}
There exists a single-pass deterministic streaming algorithm that finds a $(1-\epsilon)$ approximation of the $\LCS$ between $\of$ and $\on$ using $\tilde{O}(\sqrt{n}/\epsilon)$ memory.
\end{theorem}
\begin{proof}
By setting $\epsilon^*=\epsilon/\sqrt{n}$, Theorem \ref{thm:m2p} immediately gives us an algorithm with the approximation ratio of 
\begin{align*}
(1-\epsilon^*) ^{\sqrt{n}} \ge 1- \epsilon^* \cdot \sqrt{n}  = 1- \epsilon \,.
\end{align*}
Also, the memory of this algorithm is bounded by
\begin{align*}
\tilde{O}(\sqrt{n}+\log_{1+\epsilon^*} n) = \tilde{O}(\sqrt{n}/ \epsilon) \,.
\end{align*}
\end{proof}

%% file: ed-rootn.tex
In this section, we design a streaming algorithm that finds a $(1+\epsilon)$ approximation of the edit distance for an arbitrary $\epsilon>0$. The memory of our algorithm is $\tilde{O}(\sqrt{n}/\epsilon)$. Our algorithm is inspired by the algorithm of \cite{hss19} for approximating the edit distane in the Massively Parallel Compution (MPC) model.

Suppose that we are given a distance $d$, and we want to verify whether the edit distance between $\of$ and $\on$ is close to $d$ or not. If we can solve this subproblem, we can also find an approximation of the edit distance between $\of$ and $\on$. In order to do that, we can run the algorithm for different values of $d$ in $\{1,\lfloor (1+\epsilon) \rfloor,\lfloor (1+\epsilon)^2 \rfloor, \cdots \}$ and return the minimum $d$ that our algorithm verifies it is close to the edit distance between $\of$ and $\on$. The number of guesses for $d$ is also bounded by $O(\log_{1+\epsilon}(n))$ and we can run the algorithm for all different guesses of $d$ in parallel and return the best answer. Thus, our goal in the rest of the section is to design a streaming algorithm that given an approximate size of the edit distance, verifies whether a solution with that size exists. 

Similar to our algorithm for $\LCS$, we divide the online string into $\sqrt{n}$ windows of size $\sqrt{n}$. For simplicity and without loss of generality, we assume that the length of the string is divisible by $\sqrt{n}$ (Otherwise we can pad both online and offline strings with the same character which is not in $\Sigma$ and this does not change the edit distance). Our algorithm divides the online string into $\sqrt{n}$ windows $\on^*_1, \on^*_2, \cdots, \on^*_{\sqrt{n}}$ where $\on^*_i$ is the substring $\on[(i-1)\sqrt{n}+1, i \sqrt{n}]$ of the online string. Let assume that in the optimal solution of the edit distance, window $\on^*_1$ is mapped to the substring $\of[l_i,r_i)$. For each window $\on^*_i$, our algorithm finds a set of candidate intervals for the mapping of the this window. Roughly speaking, we show that our candidate set always contains an interval which is very close to $[l_r,r_i)$. We then show that using these intervals we can get a good approximation of the edit distance.

\paragraph{Finding Candidate Intervals.}

Consider a window $\on^*_i$ of the online string. Let us suppose that in the optimal solution, it is mapped to the substring $\of[l_i, r_i)$. We can always assume that $l_i=r_{i-1}$ for $i>1$, $l_1=1$ and $r_{\sqrt{n}}=n+1$. We then have

\begin{align*}
\ed(\on,\of) = \sum_{i=1}^{\sqrt{n}} \ed(\on^*_i, \of[l_i, r_i) ) \,.
\end{align*}

Our goal is to find a set of candidate intervals for $\on^*_i$ such that at least one of these intervals is very close to $[l_i, r_i)$. In order to design our algorithm, we first explore some properties of the interval $[l_i,r_i)$. We use $\alpha_i= (i-1) \sqrt{n}+1$, and $\beta_i= i \sqrt{n}$ to denote the starting and the ending of the window $\on^*_i$ respectively. Therefore, we have $\on^*_i= \on[\alpha_i, \beta_i]$. Recall that we have assumed that we are given a bound $d$ on the size of the edit distance. Therefore, in the optimal mapping $\on^*_i$ is mapped to a substring with the distance of at most $d$, and we must have 
\begin{align*}
| r_i -1 - \beta_i | \le d \,.
\end{align*}

It follows that $r_i \in [\beta_i+1 - d, \beta_i+1+d]$. Let $\kappa= \lfloor d \cdot \epsilon /\sqrt{n} \rfloor$. The algorithm considers all intervals with the ending point in $[\beta_i+1 - 2d, \beta_i+1+2d]$ such that the ending points are divisible by $\kappa$ (see Figure \ref{figs:gaps}).  We call these intervals, \textit{candidate intervals} and we call their endings \textit{candidate endings} . We also consider all intervals ending in $1$, i.e. intervals $[l,1)$, as candidate intervals if $1 \in [\beta_i+1 - 2d, \beta_i+1+2d]$.

\input{figs/gaps}

\begin{algorithm} [h]
 \KwData{An offline string $\of$ of length $n$, an online string $\on$, a bound $d$ for the edit distance, and an $\epsilon>0$.}
 \begin{algorithmic} [1]
 \STATE $\kappa = \lfloor d \cdot \epsilon /\sqrt{n} \rfloor$.
 \STATE $D \leftarrow$ an empty function.
 \STATE $D[1] \leftarrow 0$.

 \STATE Divide $\on$ into $\sqrt{n}$ windows $\on^*_1, \on^*_2, \ldots, \on^*_{\sqrt{n}}$ of size $\sqrt{n}$.
 
 \FOR {each window $\on^*_i$}
 	\STATE $T \leftarrow$ an empty function.
 	\STATE $\beta= i \sqrt{n}$.
 	\FOR {every integer $r$ in $[\beta+1 - 2d, \beta+1+2d]$ such that $r$ is $1$ or is divisible by $\kappa$}
 	\IF {$r$ is at least $1$ and at most $n+1$}
 		\STATE $T[r] \leftarrow \infty$.
 		\FOR{each $l \in D$ such that $l\le r$}
 			\STATE $T[r] \leftarrow \min\big\{T[r], D[l] + \ed(\of[l,r),\on^*_i) \big\}$.
 		\ENDFOR
 		\ENDIF
 	\ENDFOR
 	\STATE $D \leftarrow T$.
 \ENDFOR
 \RETURN $\min_{r \in D} D[r] + (n-r +1)$.
 \end{algorithmic}
\caption{Algorithm $\edrootn$ for approximating the edit distance.}
\label{alg:candidate}
\end{algorithm}

Our algorithm uses the dynamic programming to find the  best mapping of the $\on^*_i$ windows to their candidate intervals. Define the function $D_i$ as follows. Let $D_i[r]$ be the best mapping of the first $i$ windows to their candidate intervals such that  $\on^*_i$ is mapped to an interval ending in $r$. Note that for all candidate intervals for the window $\on^*_i$, their ending points are either $1$ or an integer in $[\beta_i+1 - 2d, \beta_i+1+2d]$ that is divisible by $\kappa$.
Therefore the number of possible different end points for the candidate intervals is bounded by $O(d / \kappa) = O(\sqrt{n}/\epsilon)$. Thus, function $D_i$ only takes $O(\sqrt{n}/\epsilon)$ values and we can store all values for this function in a memory of $\tilde{O}(\sqrt{n}/\epsilon)$. We say that $r \in D_i$, if the function $D_i$ takes the value $r$. In other words, $r$ is an end point for at least one of the candidate intervals for $\on^*_i$. 
Consider an ending point $r \in D_i$. Consider the optimal solution for $D_i[r]$. Let assume in that solution window $\on^*_i$ is mapped to an interval $[l,r)$ of the offline string. Then, the first $i-1$ windows are mapped to the substring $\of[1,l)$. Also, $\on^*_{i-1}$ is mapped to an interval with the ending point equal to $l$. Therefore $l$ is a candidate ending for $\on^*_{i-1}$. By the definition of the $D_i$ functions, $D_{i-1}[l]$ denotes the best mapping for the first $i-1$ windows such that $\on^*_{i-1}$ is mapped to an interval with the ending point equal to $l$. Thus, we have
 $$D_i[r]= 
D_{i-1}[l]+ \ed\big(\of[l,r),\on^*_i\big) \,.$$

According to the equation above, we can find the value for function $D_i$ by only using the values of function $D_{i-1}$. As we mentioned earlier, we can store the values of functions $D_{i}$ and $D_{i-1}$ in a memory of $\tilde{O}(\sqrt{n}/\epsilon)$.

Our algorithm (formally as Algorithm \ref{alg:candidate}), divides the online string into $\sqrt{n}$ windows $\on^*_1, \on^*_2, \cdots, \on^*_{\sqrt{n}}$. It also keeps a function $D$ of size at most $O(\sqrt{n}/\epsilon)$ which represents values of the function $D_i$ after arrival of the window $\on^*_i$. Upon arrival a new window $\on^*_{i+1}$, the algorithm stores all characters of $\on^*_{i+1}$ in the memory and update the function $D$ based on the update rule below.

\begin{align}
\label{eq:secondedupdate}
D_i[r]= \min_{ l \in D_{i-1} } D_{i-1}[l]+ \ed\big(\of[l,r),\on^*_i\big) \,.
\end{align}
According to what we have discussed, the update rule above gives the optimal answer for each $D_i$. 

\begin{theorem}\label{th:sar}
Algorithm \ref{alg:candidate} uses $\tilde{O}(\sqrt{n}/\epsilon)$ memory and finds $(1+5\epsilon)$ approximation of the edit distance between $\of$ and $\on$.
\end{theorem}
\begin{proof}
Consider an optimal solution for $\ed(\of,\on)$. Let $\opt$ be the size of this solution, and $d$ be the best guess of our algorithm for the edit distance between $\of$ and $\on$. Then, we have
\begin{align}
\label{eq:bestguess}
\opt \le d \le (1+\epsilon) \opt \,.
\end{align}
 Suppose in the optimal solution,  window $\on^*_i$ is mapped to the substring $\of[l_i,r_i)$ of the offline string. Then we have
\begin{align}
\label{eq:edbestmap}
\opt = \ed(\of,\on) = \sum_{i=1}^{\sqrt{n}} \ed\big(\of[l_i,r_i), \on^*_i\big) \,.
\end{align}
We also have that $l_1=1$ and $l_i = r_{i-1}$ for $i >1$. Also, $r_{\sqrt{n}}=n+1$. Let $\mathcal{C}$ be the set of all integers such that they can be a candidate ending point for one of $\on^*_i$ windows. In other words,
\begin{align*}
\mathcal{C} = \{1, \kappa, 2\kappa, \cdots, \lfloor n/\kappa \rfloor \kappa \} \,.
\end{align*}
For each $l_i$ (respectively, $r_i$), let $l'_i$ (resp., $r'_i$) be the largest number in $\mathcal{C}$ that is at most $l_i$ (resp., $r_i$). Then, for each $l_i$, we have
\begin{align}
\label{eq:lrb}
l_i - \kappa < l'_i \le l_i \,.
\end{align}
Similarly, for each $r_i$, we have
\begin{align}
\label{eq:rrb}
r_i - \kappa < r'_i \le r_i \,.
\end{align}
It follows that for each interval $[l'_i, r'_i)$ we have
\begin{align}
\label{eq:ned}
\ed\big(\of[l'_i,r'_i), \on^*_i\big) &\le \ed\big(\of[l'_i,r'_i), \of[l_i,r_i)\big)+\ed\big(\of[l_i,r_i), \on^*_i\big) \nonumber & \text{By the triangle inequality.} \\
& \le 2\kappa + \ed\big(\of[l_i,r_i), \on^*_i\big) \,. & \text{By (\ref{eq:lrb}) and (\ref{eq:rrb}).}
\end{align}
It follows from (\ref{eq:ned}) that
\begin{align}
\label{eq:nrm}
&\sum_{i=1}^{\sqrt{n}} \ed\big(\of[l'_i,r'_i), \on^*_i\big) \nonumber
\\& \le 2\kappa \cdot \sqrt{n} + \sum_{i=1}^{\sqrt{n}} \ed\big(\of[l_i,r_i), \on^*_i\big) \nonumber \\
& \le 2 \epsilon \cdot d + \sum_{i=1}^{\sqrt{n}} \ed\big(\of[l_i,r_i), \on^*_i\big) \nonumber & \text{Since $\kappa =  \lfloor d \cdot \epsilon /\sqrt{n} \rfloor$.}  \\
& = 2 \epsilon \cdot d + \opt \nonumber & \text{By (\ref{eq:edbestmap}).} \\
& \le \opt (1+ 2 \epsilon (1+\epsilon)) \nonumber & \text{By (\ref{eq:bestguess}).} \\
& \le \opt (1+3\epsilon) \,. 
\end{align}
Therefore, the size of the solution that maps each window $\on^*_i$ to $\of[l'_i,r'_i]$ is at most $(1+3 \epsilon)\opt$. We show that our algorithm almost finds this solution. We claim that each $r'_i$ is a candidate endpoint for $\on^*_i$. Since $r'_i \in \mathcal{C}$, it is either 1 or it is divisible by $\kappa$. To show that $r'_i$ can be the end point of some candidate interval for $\on^*_i$, it is sufficient to show that $r'_i$ is in $[\beta_i+1 - 2d, \beta_i+1+2d]$ where $\beta_i= i \sqrt{n}$ is the end point of the window $\on^*_i$.

Because the size of the edit distance between $\of$ and $\on$ is bounded by $d$, we have
\begin{align*}
|r_i - \beta_i+1| \le d \,.
\end{align*} 
This along with (\ref{eq:rrb}) implies that
\begin{align*}
|r'_i - \beta_i+1| \le |r'_i - r_i| +|r_i - \beta_i+1| \le \kappa + d \le 2d \,.
\end{align*}
Therefore, $r'_i$ is in $[\beta_i+1 - 2d, \beta_i+1+2d]$ and $[l'_i,r'_i)$ is a candidate interval for $\on^*_i$. Thus in this solution, every window $\on^*_i$ is mapped to one of its candidate intervals. Consider the last window, it is mapped to the interval $[l'_{\sqrt{n}},r'_{\sqrt{n}})$. Let $q=r'_{\sqrt{n}}$. By the definition of $D_i$ functions, $D_{\sqrt{n}}[q]$ is the cost of the best solution such that each window is mapped to one of its candidate interval, and the ending of the last interval is $q$. Therefore,

\begin{align}
\label{eq:dnapx}
D_{\sqrt{n}}[q] &\le \sum_{i=1}^{\sqrt{n}} \ed\big(\of[l'_i,r'_i), \on^*_i\big) \nonumber \\
&\le (1+3\epsilon) \opt \,. & \text{By (\ref{eq:nrm}).}
\end{align}
After arrival of all windows, in  Algorithm \ref{alg:candidate} function $D$ will be equal to $D_{\sqrt{n}}$, and the algorithm returns
$\min_{r \in D} D[r] + (n-r +1)$ where $(n-r+1)$ is the edit distance of between part of offline string that is not in the mapping represented by $l'_i$'s and $r'_i$'s. Therefore, the solution of the algorithm is bounded by
\begin{align*}
\min_{r \in D} D[r] + (n-r +1) &= \min_{r \in D_{\sqrt{n}}} D_{\sqrt{n}}[r] + (n-r +1) \\
&\le D_{\sqrt{n}}[q] + (n-q+1) \\
&\le (1+3\epsilon)\opt + (n-q+1)  & \text{By (\ref{eq:dnapx}).} \\
&\le (1+3\epsilon)\opt + (r_{\sqrt{n}}-q) & \text{Since $r_{\sqrt{n}}= n+1$.} \\
&\le (1+3\epsilon) \opt + \kappa & \text{ By (\ref{eq:rrb}).} \\
&\le (1+3\epsilon) \opt + \epsilon \cdot d \\
&\le (1+5\epsilon) \opt \,.
\end{align*}
Therefore the approximation ratio of the algorithm is bounded by $(1+5 \epsilon)$ and it proves the theorem.
\end{proof}

The above theorem immediately implies the following.

\begin{theorem}
\label{thm:mc2}
For any $\epsilon>0$, there exists a single-pass deterministic streaming algorithm that finds a $(1+\epsilon)$ approximation of the edit distance between $\on$ and $\of$ using $\tilde{O}(\sqrt{n}/\epsilon)$ memory.
\end{theorem}

%% file: figs/gaps.tex
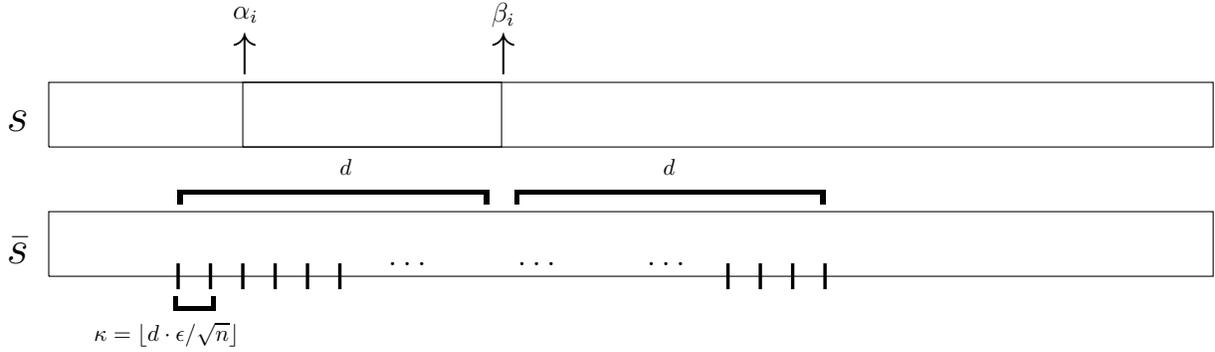
\begin{figure}[!h]
\usetikzlibrary{patterns}

\begin{center}
\begin{tikzpicture}[scale=0.86, transform shape]

\draw (0,0) -- (18,0) -- (18,1) -- (0,1) -- (0,0);

\draw (3,0) -- (7,0) -- (7,1) -- (3,1) -- (3,0);

\draw (0,-2) -- (18,-2) -- (18,-1) -- (0,-1) -- (0,-2);

\node[text width=4.8cm] at (4.4,-0.8) 
{   \Large $\overbracket{\hspace{4.8cm}}$ };

\node[text width=1cm] at (5,-0.3) 
{   \small $d$ };

\node[text width=4.8cm] at (9.6,-0.8) 
{   \Large $\overbracket{\hspace{4.8cm}}$ };

\node[text width=1cm] at (10,-0.3) 
{   \small $d$ };

\node[text width=0.65cm] at (2.25,-2.4) 
{   \Large $\underbracket{\hspace{0.65cm}}$ };

\node[text width=5cm] at (3.2,-2.9) 
{   \small $\kappa = \lfloor d \cdot \epsilon / \sqrt{n} \rfloor$ };

\draw [line width=0.45mm ](2,-2.2) -- (2,-1.8);
\draw [line width=0.45mm ] (2.5,-2.2) -- (2.5,-1.8);
\draw [line width=0.45mm ] (3,-2.2) -- (3,-1.8);
\draw [line width=0.45mm ] (3.5,-2.2) -- (3.5,-1.8);
\draw [line width=0.45mm ] (4,-2.2) -- (4,-1.8);
\draw [line width=0.45mm ] (4.5,-2.2) -- (4.5,-1.8);

\node[text width=0.5cm] at (5.5,-1.8) 
{   \Large $\ldots$ };
\node[text width=0.5cm] at (7.5,-1.8) 
{   \Large $\ldots$ };
\node[text width=0.5cm] at (9.5,-1.8) 
{   \Large $\ldots$ };

\draw [line width=0.45mm ] (10.5,-2.2) -- (10.5,-1.8);
\draw [line width=0.45mm ] (11,-2.2) -- (11,-1.8);
\draw [line width=0.45mm ] (11.5,-2.2) -- (11.5,-1.8);
\draw [line width=0.45mm ] (12,-2.2) -- (12,-1.8);

\node[text width=0.5cm] at (-0.40,-1.6) 
{   \huge $\of$ };

\node[text width=0.5cm] at (-0.40,0.4) 
{   \huge $\on$ };

\node[text width=0.3cm] at (3,1.45) 
{   \huge $\uparrow$ };
\node[text width=0.3cm] at (3,2.05) 
{   \large $\alpha_i$ };

\node[text width=0.3cm] at (7,1.45) 
{   \huge $\uparrow$ };
\node[text width=0.3cm] at (7,2.05) 
{   \large $\beta_i$ };

\end{tikzpicture}
\end{center}
\caption{The locations of the ending points for potential intervals of a window are illustrated in this figure. Thick segments show the ending points.}\label{figs:gaps}
\end{figure}

%% file: omitted.tex
\logtwo*
\begin{proof}
It is known that the $\LCS$ and the edit distance are in \textit{Nondeterministic Logarithmic-space} (NL) complexity class. This means that we can solve these problem using a non-deterministic Turing machine with a memory of $O(\log n)$. Savitch's theorem \cite{savitch1970relationships} says that every problem in NL can be solved using a deterministic Turing matching with a memory of $O(\log^2 n)$, which implies the theorem.
\end{proof}